\documentclass[runningheads,11points]{llncs}

\usepackage[utf8x]{inputenc}
\usepackage{amsmath,ifthen,color,eurosym}

\usepackage{wrapfig,hyperref,cite}
\usepackage{latexsym,amsmath,amssymb,url}
\usepackage{fancyhdr}
\usepackage{graphicx}
\usepackage{stmaryrd,wasysym}
\usepackage{colortbl,datetime}
\usepackage{boxedminipage}
\usepackage{slashbox,xcolor}
\definecolor{gray1}{gray}{0.775}
\definecolor{gray2}{gray}{0.75}

\usepackage[lined,boxed,commentsnumbered]{algorithm2e}

\newcommand{\remove}[1]{}

\spnewtheorem{observation}{Observation}{\bfseries}{\itshape}
\spnewtheorem{claimN}{Claim}{\bfseries}{\itshape}

\spnewtheorem{remarkLI}{Remark}{\bfseries}{\itshape}

\newcommand{\tw}{\mbox{\bf tw}}

\renewenvironment{proof}{\noindent \textsc{Proof:}}{\hfill$\square$\medskip}

\usepackage{tikz}

\newcommand{\bs}{-}
\newcommand{\sm}{\setminus}
\newcommand{\ls}{L_{sel}}
\newcommand{\rs}{R_{sel}}

\newcommand{\mA}{\mathcal{A}}
\newcommand{\PP}{\mathcal{P}}
\newcommand{\OO}{\mathcal{O}}
\newcommand{\QQ}{\mathcal{Q}}

\newcommand{\mR}{\mathcal{R}}
\newcommand{\mL}{\mathcal{L}}

\newcommand{\paraprobl}[4]
{
  \begin{flushleft}
    \fbox{
      \begin{minipage}{11.8cm}
        \noindent {\textsc {#1}}\\
        {\bf Input:} #2\\
        {\bf Parameter:} #4\\
        {\bf Output:} #3
      \end{minipage}
    }
  \end{flushleft}
}

\newcommand{\capt}[2]{\draw[white] (#1) -- node[above,sloped,black] {#2} +(2,0)}

\newcommand{\add}[1]{\textcolor{blue}{#1}}

\newcommand{\supprOK}[1]{}
\newcommand{\addOK}[1]{#1}

\title{Parameterized Complexity Dichotomy for $(r,\ell)$-\textsc{Vertex Deletion}\thanks{This work was partially supported by CNPq, CAPES, FAPERJ, and COFECUB.}
}
\author{Julien Baste~\inst{1}, Luerbio Faria~\inst{2}, Sulamita Klein~\inst{3}, and Ignasi Sau~\inst{1}}

\authorrunning{J. Baste, L. Faria, S. Klein, and I. Sau}
\titlerunning{Parameterized Complexity Dichotomy for $(r,\ell)$-\textsc{Vertex Deletion}}

\institute{AlGCo project team, CNRS, LIRMM, Montpellier, France.\\
\email{julien.baste@lirmm.fr, ignasi.sau@lirmm.fr}
\and FFP, Universidade do Estado do Rio de Janeiro, Rio de Janeiro, Brazil.\\
\email{luerbio@cos.ufrj.br}
\and Universidade Federal do Rio de Janeiro, Rio de Janeiro, Brazil.\\
\email{sula@cos.ufrj.br}}

\begin{document}

\maketitle
\setcounter{footnote}{0}


\begin{abstract}
For two integers $r, \ell \geq 0$, a graph $G = (V, E)$ is an \emph{$(r,\ell)$-graph} if $V$ can be partitioned into $r$ independent sets and $\ell$ cliques. In the parameterized $(r,\ell)$-\textsc{Vertex Deletion} problem, given a graph $G$ and an integer $k$, one has to decide whether at most $k$ vertices can be removed from $G$ to obtain an $(r,\ell)$-graph. This problem is {\sf NP}-hard if $r+\ell \geq 1$ and encompasses several relevant problems such as \textsc{Vertex Cover} and \textsc{Odd Cycle Transversal}. The parameterized complexity of $(r,\ell)$-\textsc{Vertex Deletion} was known for all values of $(r,\ell)$ except for $(2,1)$, $(1,2)$, and $(2,2)$. We prove that each of these three cases is {\sf FPT} and, furthermore, solvable in single-exponential time, which is asymptotically optimal in terms of $k$. We consider as well the version of $(r,\ell)$-\textsc{Vertex Deletion} where the set of vertices to be removed has to induce an independent set, and provide also a parameterized complexity dichotomy for this problem.

\vspace{0.19cm}
\noindent\textbf{Keywords:} graph modification problem; parameterized complexity; iterative compression; {\sf FPT}-algorithm; single-exponential algorithm.
\end{abstract}

\section{Introduction}
\label{sec:intro}
\textbf{Motivation}. Let $r, \ell \geq 0$  be two fixed integers. A graph $G = (V, E)$ is an \emph{$(r,\ell)$-graph} if $V$ can be partitioned into $r$ independent sets and $\ell$ cliques. In the parameterized $(r,\ell)$-\textsc{Vertex Deletion} problem, we are given a graph $G$ and an integer parameter $k$, and the task is to decide whether at most $k$ vertices can be removed from $G$ so that the resulting graph is an $(r,\ell)$-graph. The optimization version of this problem is known to be {\sf NP}-hard for $r + \ell \geq 1$ by a classical result of Lewis and Yannakakis~\cite{LeYa80}.  The $(r,\ell)$-\textsc{Vertex Deletion} problem has a big expressive power as it captures several relevant problems for particular cases of the pair $(r,\ell)$. Indeed, for instance, the case $(1,0)$ corresponds to \textsc{Vertex Cover}, the case $(2,0)$ to \textsc{Odd Cycle Transversal}, the case $(1,1)$ to \textsc{Split Vertex Deletion}, and the case $(3,0)$ to whether at most $k$ vertices can be removed so that the resulting graph is 3-colorable.

In this article we are interested in the parameterized complexity of $(r,\ell)$-\textsc{Vertex Deletion}; see~\cite{DF13,FG06,Nie06} for  introductory textbooks to the field. We just recall that a problem defined on an $n$-vertex graph is \emph{fixed-parameter tractable} ({\sf FPT} for short) with respect to a parameter $k$ if it can be solved in \emph{{\sf FPT}-time}, i.e., in time $f(k) \cdot n^{O(1)}$.
An {\sf FPT}-algorithm  that runs in time $2^{O(k)} \cdot n^{O(1)}$ is called \emph{single-exponential}. For the case of \textsc{Vertex Cover} (\textsc{VC} for short), a simple branching algorithm yields an {\sf FPT}-algorithm in time $2^{k} \cdot n^{O(1)}$. The currently fastest algorithm~\cite{CKX10} runs in time $1.27^k \cdot n^{O(1)}$. For \textsc{Odd Cycle Transversal} (\textsc{OCT} for short), the problem was not known to be {\sf FPT} until Reed \emph{et al}.~\cite{RSV04} introduced the celebrated technique of \emph{iterative compression} and solved \textsc{OCT} in time $3^{k} \cdot n^{O(1)}$. The current fastest algorithm~\cite{LNRRS14} uses linear programming and runs in time $2.31^{k} \cdot n^{O(1)}$.
The \textsc{Split Vertex Deletion} problem can be easily seen to be solvable in single-exponential time since split graphs can be characterized by a {\sl finite} set of forbidden induced subgraphs~\cite{FoHa77,Cai96}.
The current fastest algorithm is by Cygan and Pilipczuk~\cite{Cyg13} and runs in time $\OO(1.2738^kk^{\OO(\log k)}+n^3)$ using \textsc{Vertex Cover} as a subroutine.
It improves the previously fastest algorithm that runs in time $2^{k} \cdot n^{O(1)}$ and uses iterative compression~\cite{GK0MPRR15}, which in turn improves another algorithm using linear programming~\cite{LNRRS14} that runs in time $2.31^{k} \cdot n^{O(1)}$. (See also~\cite{KrNa13} for parameterized algorithms for $(r,\ell)$-\textsc{Vertex Deletion} on perfect graphs.)

Note that solving $(r,\ell)$-\textsc{Vertex Deletion} on a graph $G$ is equivalent to solving $(\ell,r)$-\textsc{Vertex Deletion} on the complement of $G$. This observation implies that the case $(0,2)$ can also be solved in time $2.31^{k} \cdot n^{O(1)}$. Note also that if $\max\{r,\ell\} \geq 3$, then $(r,\ell)$-\textsc{Vertex Deletion} is para-{\sf NP}-complete, hence unlikely to be {\sf FPT}, as for $k=0$ the problem corresponds to the recognition of $(r,\ell)$-graphs, which is {\sf NP}-complete if and only if $\max\{r,\ell\} \geq 3$~\cite{Bra96,FHKM03}.

Therefore, concerning the parameterized complexity of the $(r,\ell)$-\textsc{Vertex Deletion} problem on general graphs, the above discussion implies that the only open cases are $(2,1)$, $(1,2)$, and $(2,2)$. Note also that all the cases that are known to be {\sf FPT} can be solved in single-exponential time.

\vspace{.1cm}
\noindent\textbf{Our results}. In this article we prove that each of the above three open cases is {\sf FPT} and can also be solved in single-exponential time, thus completely settling the parameterized complexity of $(r,\ell)$-\textsc{Vertex Deletion}. That is, excluding the trivial case where $r + \ell = 0$, we obtain the following dichotomy: the problem is {\sf FPT} and solvable in single-exponential time  if $\max\{r,\ell\} \leq 2$, and para-{\sf NP}-complete otherwise. As discussed later, a single-exponential running time is asymptotically best possible in terms of $k$ unless the Exponential Time Hypothesis (ETH) fails. A summary of the parameterized complexity of $(r,\ell)$-\textsc{Vertex Deletion} is shown in Table~\ref{tab:VD}, where for each value of $(r,\ell)$, the name of the problem (if any), the function $f(k)$, and the appropriate references are given. We denote by $\overline{\textsc{VC}}$ and $\overline{\textsc{OCT}}$ the complementary problems of \textsc{VC} and \textsc{OCT}, respectively. The results of this article correspond to the gray boxes, `p-{\sf NP}-c' stands for `para-{\sf NP}-complete', and `{\sf P}' means that the corresponding problem is polynomial-time solvable\footnote{We would like to mention here that after  this article appeared in arXiv (\texttt{abs/1310.6205}), we learnt that Kolay and  Panolan  (\texttt{abs/1504.08120}, further published in~\cite{KP15}) obtained simultaneously and independently the same results that we present in Table~\ref{tab:VD} using very similar techniques.}.

\begin{table}[h]
\vspace{-.25cm}
\parbox{.45\linewidth}{
\begin{flushleft}
\begin{tabular}{|c||c|c|c|c|}
\hline
 &  &  &  &      \\[-.8ex]
3  & \scriptsize{p-{\sf NP}-c} & \scriptsize{p-{\sf NP}-c} & \scriptsize{p-{\sf NP}-c} & \scriptsize{p-{\sf NP}-c}     \\
 & \cite{Bra96} & \cite{Bra96} & \cite{Bra96} &  \cite{Bra96}    \\
\hline
 & $\overline{\textsc{OCT}}$ & \cellcolor{gray1} & \cellcolor{gray1} &      \\
2  & $2.31^{k}$ & $\cellcolor{gray1}3.31^k$ & $\cellcolor{gray1}3.31^k$ &  \scriptsize{p-{\sf NP}-c}    \\
 & \cite{LNRRS14}  & \cellcolor{gray1}Coro~\ref{coro:vd} & \cellcolor{gray1}Thm~\ref{lemma:dttvd} &    \cite{Bra96}  \\
\hline
 & $\overline{\textsc{VC}}$ & \textsc{Split D.} & \cellcolor{gray1}  &      \\
1  & $1.27^{k}$ & $2^k$ & $\cellcolor{gray1}3.31^k$ & \scriptsize{p-{\sf NP}-c}     \\
 & \cite{CKX10} & \cite{GK0MPRR15} & \cellcolor{gray1}Coro~\ref{coro:vd} &  \cite{Bra96}    \\
 \hline
 &  & \textsc{VC} & \textsc{OCT} &      \\
0  & {\sf P}  & $1.27^{k}$ & $2.31^{k}$ & \scriptsize{p-{\sf NP}-c}     \\
 & \scriptsize{trivial} & \cite{CKX10} & \cite{LNRRS14} &  \cite{Bra96}    \\
\hline
\hline
\tiny{\slashbox{\small{$\ell$}}{\small{$r$}}} &  0 & 1  & 2 & 3    \\
\hline

\end{tabular}
\end{flushleft}
\caption{\label{tab:VD}Summary of results for the \textsc{$(r,\ell)$-Vertex Deletion} problem. Our results correspond to gray cells.}
}
\hspace{1.25cm}
\parbox{.45\linewidth}{
\begin{center}
\begin{tabular}{|c||c|c|c|c|}
\hline
 &  &  &  &      \\[-.8ex]
3  & \scriptsize{p-{\sf NP}-c} & \scriptsize{p-{\sf NP}-c} & \scriptsize{p-{\sf NP}-c} & \scriptsize{p-{\sf NP}-c}     \\
 & \cite{Bra96} & \cite{Bra96} & \cite{Bra96} &  \cite{Bra96}    \\
\hline
 & \cellcolor{gray1} & \cellcolor{gray1} & \cellcolor{gray1} &      \\[-1.2ex]
2  & \cellcolor{gray1}{\sf P} & \cellcolor{gray1}{\sf P} & \cellcolor{gray1}$2^{2^{O(k^2)}}$ &  \scriptsize{p-{\sf NP}-c}    \\
 & \cellcolor{gray1}Thm~\ref{thm:easy}  & \cellcolor{gray1}Thm~\ref{thm:easy} & \cellcolor{gray1}Thm~\ref{theorem:ittvd} &    \cite{Bra96}  \\
\hline
 & \cellcolor{gray1} & \cellcolor{gray1} & \cellcolor{gray1}  &      \\[-1.2ex]
1  & \cellcolor{gray1}{\sf P} & \cellcolor{gray1}{\sf P} & \cellcolor{gray1}$2^{2^{O(k^2)}}$ & \scriptsize{p-{\sf NP}-c}     \\
 & \cellcolor{gray1}Thm~\ref{thm:easy} & \cellcolor{gray1}Thm~\ref{thm:easy} & \cellcolor{gray1}Coro~\ref{theorem:itovd} &  \cite{Bra96}    \\
 \hline
 &  & \cellcolor{gray1}\textsc{IVC} & \textsc{IOCT} &      \\
0  & {\sf P}  & \cellcolor{gray1}{\sf P} & $2^{2^{O(k^2)}}$ & \scriptsize{p-{\sf NP}-c}     \\
 & \scriptsize{trivial} & \cellcolor{gray1}Thm~\ref{thm:easy} & \cite{MOR10} &  \cite{Bra96}    \\
\hline
\hline
\tiny{\slashbox{\small{$\ell$}}{\small{$r$}}} &  0 & 1  & 2 & 3    \\
\hline
\end{tabular}
\end{center}
\caption{\label{tab:IVD}Results for \textsc{Independent $(r,\ell)$-Vertex Deletion}. Our results correspond to gray cells.}
}
\vspace{-.3cm}
\end{table}

We also consider the version of $(r,\ell)$-\textsc{Vertex Deletion} where the set $S$ of at most $k$ vertices to be removed has to further satisfy that $G[S]$ is an independent set. We call this problem \textsc{Independent $(r,\ell)$-Vertex Deletion}. Note that, in contrast to \textsc{$(r,\ell)$-Vertex Deletion}, the cases $(r,\ell)$ and $(\ell,r)$ may not be symmetric anymore. This problem has received little attention in the literature and, excluding the most simple cases, to the best of our knowledge only the case $(2,0)$ has been studied by Marx \emph{et al}.~\cite{MOR10}, who proved it to be {\sf FPT}.  Similarly to \textsc{$(r,\ell)$-Vertex Deletion},  the problem is para-{\sf NP}-complete if $\max\{r,\ell\} \geq 3$. As an additional motivation for studying this problem, note that solving \textsc{Independent $(r,\ell)$-Vertex Deletion} on an input $(G,k)$ corresponds exactly to deciding whether $G$ is an $(r+1,\ell)$-graph where one of the independent sets has size at most $k$.

We manage to provide a complete characterization of the parameterized complexity of \textsc{Independent $(r,\ell)$-Vertex Deletion}. The complexity landscape turns out to be richer than the one for \textsc{$(r,\ell)$-Vertex Deletion}, and one should rather speak about a trichotomy: the problem is polynomial-time solvable if $r\leq 1$ and $\ell\leq 2$, {\sf NP}-hard and {\sf FPT} if $r = 2$ and $\ell \leq 2$, and para-{\sf NP}-complete otherwise. In particular, as discussed at the end of the previous paragraph, it follows from our results that for $\ell \in \{0,1,2\}$, the recognition of the class of $(3,\ell)$-graphs such that one of the independent sets has size at most $k$ is in {\sf FPT} with parameter $k$. A summary of the complexity of \textsc{Independent $(r,\ell)$-Vertex Deletion} is shown in Table~\ref{tab:IVD}, where our results correspond  to the gray boxes. We would like to note that some of the polynomial cases, such as the case $(1,0)$, are not difficult to prove and may be already known, although we are not aware of it.

\vspace{.1cm}
\textbf{Our techniques}.
As most of the previous work mentioned before, our algorithms for \textsc{$(r,\ell)$-Vertex Deletion} (Section~\ref{sec:rlvd}) are based on iterative compression. We provide an algorithm for \textsc{$(2,2)$-Vertex Deletion}, and we show that  \textsc{$(1,2)$-Vertex Deletion} and \textsc{$(2,1)$-Vertex Deletion} can be easily reduced to \textsc{$(2,2)$-Vertex Deletion}.
For completeness, we include in Section~\ref{ap:iterative-compression} some well-known properties of iterative compression. As a crucial ingredient in our algorithms, we prove (Lemma~\ref{lemma:rlpartition} in Section~\ref{sec:prelim})
\supprOK{that an $n$-vertex $(r,\ell)$-graph has at most $(n+1)^{2 r  \ell}$  distinct $(r,\ell)$-partitions,}\addOK{that given two $(r,\ell)$-partitions of an $n$-vertex $(r,\ell)$-graph, these two $(r,\ell)$-partitions differ by at most $2r\ell$ vertices,}
where an  \emph{$(r,\ell)$-partition} of an $(r,\ell)$-graph $G$ is a partition $(R, L)$ of $V(G)$ such that $G[R]$ is an $(r,0)$-graph and $G[L]$ is a $(0,\ell)$-graph.
\supprOK{Furthermore, all these partitions can be generated in polynomial time if $\max\{r, \ell\} \leq 2$.}\addOK{Furthermore, if $\max\{r, \ell\} \leq 2$, it is known that we can find an $(r, \ell)$-partition of an $(r,\ell)$-graph in polynomial time~\cite{Bra96}.} This implies, in particular, that an $n$-vertex $(r,\ell)$-graph has at most $(n+1)^{2 r  \ell}$  distinct $(r,\ell)$-partitions that can be generated in polynomial time if $\max\{r, \ell\} \leq 2$.
 This result generalizes the fact that a split graph has at most $n + 1$ split partitions~\cite{Golumbic04}, which was used in the algorithms of~\cite{GK0MPRR15}.

Our algorithms for \textsc{Independent $(r,\ell)$-Vertex Deletion} (Section~\ref{sec:irlvd}) are slightly more involved, and do not explicitly use iterative compression. Again, we provide an algorithm for \textsc{Independent $(2,2)$-Vertex Deletion} and then show that \textsc{Independent $(2,1)$-Vertex Deletion} can be reduced to \textsc{Independent $(2,2)$-Vertex Deletion}. We make use our algorithms for \textsc{$(2,2)$-Vertex Deletion}  to obtain a set of vertices $S$ that allows us to exploit the structure of $G-S$. A crucial ingredient here is the {\sf FPT}-algorithm of Marx \emph{et al}.~\cite{MOR10} to solve the \textsc{Restricted Independent OCT} problem (see Section~\ref{sec:prelim} for the definition).

\vspace{.1cm}
\textbf{Remarks and further research}. Having completely settled the parameterized complexity of  \textsc{$(r,\ell)$-Vertex Deletion} and  \textsc{Independent $(r,\ell)$-Vertex Deletion}, a natural direction is to improve the running times of our algorithms. We did not focus in this article on optimizing the degree of the polynomial $n^{O(1)}$ involved in our running times. Concerning the function $f(k)$, for \textsc{$(r,\ell)$-Vertex Deletion} this improvement would be possible, under ETH, only in the basis of the function $3.31^k$ (see Theorem~\ref{theorem:nsevd}). For \textsc{Independent $(r,\ell)$-Vertex Deletion}, there may be room for improvement in the function $2^{2^{O(k^2)}}$ that we obtain mainly by analyzing the running time of the algorithm of Marx \emph{et al}.~\cite{MOR10} to solve \textsc{Restricted Independent OCT}, which was not explicit in their article.

Concerning the existence of polynomial kernels for \textsc{$(r,\ell)$-Vertex Deletion}, a challenging research avenue is to apply the techniques used by
Kratsch and Wahlstr{\"{o}}m~\cite{KrWa14} for obtaining a randomized polynomial kernel for \textsc{OCT} to the cases $(2,1)$, $(1,2)$, and $(2,2)$, or to prove that these problems do not admit polynomial kernels. The ideas for the case $(1,1)$ may also be helpful~\cite{GK0MPRR15}.

Finally, it is worth mentioning that if the input graph is restricted to be \emph{planar}, there exists a randomized subexponential algorithm for \textsc{OCT}~\cite{LSW12} running in time $O(n^{O(1)} + 2^{O(\sqrt{k} \log k)}n)$. As in a planar graph any  clique is of size at most $4$, by guessing one or two cliques and then applying this algorithm, we obtain randomized algorithms in time $2^{O(\sqrt{k} \log k)}\cdot n^{O(1)}$ for
\textsc{$(2,1)$-Vertex Deletion}, {\textsc{$(1,2)$-Vertex Deletion}}, and {\textsc{$(2,2)$-Vertex Deletion}} on planar graphs.

\section{Preliminaries}
\label{sec:prelim}

We use standard graph-theoretic notation, and  the reader is referred to~\cite{Diestel05} for any undefined term.
All the graphs we consider are undirected and contain neither loops nor multiple edges. If $S \subseteq V(G)$, we define $G-S = G[V(G) \setminus S]$.
The \emph{complement} of a graph $G=(V,E)$ is denoted by \emph{$\overline{G}$}, that is, $\overline{G} = (V,E')$ with $E' = \{\{x,y\}\in (V\times V)\sm E \}$.
Throughout the article $n$ denotes the number of vertices of the input graph of the problem under consideration.

It is shown in~\cite{Golumbic04} that a $(1,1)$-graph has at most $n+1$ distinct $(1,1)$-partitions. We generalize this property in the following lemma, whose proof is based on the proof of~\cite[Theorem 3.1]{FHKM03}.



\begin{lemma}
\label{lemma:rlpartition}
Let $r$ and $\ell$ be two fixed integers, and let $(R,L)$ and $(R', L')$ be two $(r,\ell)$-partitions of a graph $G$.
Let $\ls= L' \cap R$ and $\rs = R' \cap L$.
Then $\ls$ and $\rs$ are both of size at most $r\ell$, $R' = (R \sm \ls) \cup \rs$, and $L' = (L \sm \rs) \cup \ls$.
\end{lemma}

\begin{proof}  Let $G=(V,E)$ be an $(r,\ell)$-graph,
and let $(R, L)$ and $(R', L')$ be two distinct $(r, \ell)$-partitions of $G$.
We claim that $|R \cap L'| \leq r \ell$.
Indeed, assume that there exists a set $S$ of $r \ell + 1$ vertices in $R \cap L'$.
As $S \subseteq L'$, by the pigeonhole principle there exists a subset $S' \subseteq S$ of size $r+1$ such that $G[S']$ is a clique.
As $S'\subseteq R$,  the $(r,0)$-graph $G[R]$ contains a clique $G[S']$ of size $r+1$, that contradict the definition of an $(r,0)$-graph.
Symmetrically, it also holds that  $|R' \cap L| \leq r \ell$, and the lemma follows.
\end{proof}


For our algorithms we need the following restricted versions of \textsc{OCT}.

\paraprobl
{\textsc{Restricted Odd Cycle Transversal} (\textsc{Restricted OCT})}
{A graph $G=(V,E)$, a set $D \subseteq V$, and an integer $k$.}
{A set $S \subseteq D$ of size at most $k$ such that $G \bs S$ is bipartite or a correct report that such a set does not exist. }
{$k$.}

\vspace{-.4cm}
\paraprobl
{\textsc{Restricted Independent Odd Cycle Transversal}}
{A graph $G=(V,E)$, a set $D \subseteq V$, and an integer $k$.}
{An independent set $S \subseteq D$ of size at most $k$ such that $G \bs S$ is bipartite or a correct report that such a set does not exist. }
{$k$.}

\begin{lemma}
\label{lemma:roct}
\textsc{Restricted OCT} can be solved in time $2.31^{k}\cdot n^{O(1)}$.
\end{lemma}

\begin{proof}
The algorithm from Lokshtanov \emph{et al.}~\cite{LNRRS14} solves \textsc{OCT} in time $2.31^{k}\cdot n^{O(1)}$.
For our lemma, we use this algorithm on a modified input. Let $G = (V,E)$ be a graph,  $D \subseteq V$, and  $k$ an integer.
We want to solve \textsc{Restricted OCT} on $(G,D,k)$.
Let $G'=(V',E')$, where $V' = D \cup \{v_i : v \in V\sm D, i \in \{0,\ldots, k\}\}$ and $E' = (E \cap (D\times D)) \cup \{\{v_i,w\} : v \in V\sm D, i \in \{0,\ldots, k\}, w \in D, \{v,w\}\in E\} \cup \{\{v_i,w_j\} : v,w \in V \setminus D, i,j \in \{0,\ldots, k\}, \{v,w\}\in E\}$.
That is, for each vertex $v$ not in $D$, we make $k+1$ copies of $v$ with the same neighborhood as  $v$, making its choice for the solution impossible. Then we solve \textsc{Odd Cycle Transversal} on $(G',k)$, giving us a solution of \textsc{Restricted OCT} on $(G,D,k)$.
\end{proof}

By looking carefully at the proof of \cite[Theorem 4.3]{MOR10}, we have the following theorem. We will analyze the running time of the algorithm in Subsection~\ref{sec:time}.

\begin{theorem}[Marx \emph{et al}.~\cite{MOR10}]
\label{theorem:roct}
\textsc{Restricted Independent OCT} is {\sf FPT}.
\end{theorem}

We will also need to deal with the {\textsc{Independent Vertex Cover}  problem, which given a graph $G$ and an integer $k$, asks whether $G$ contains a set $S \subseteq V(G)$ of size at most $k$ that is both a vertex cover of $G$ and an independent set.


\begin{lemma}
\label{lemma:ivc}
{\textsc{Independent Vertex Cover}} can be solved in linear time.
\end{lemma}
\begin{proof}
Let $G$ be a graph and let $k$ be a positive integer. Note that for \textsc{Independent Vertex Cover} to admit a solution in $G$, in particular $G$ needs to be 2-colorable. Hence, if $G$ is not bipartite, we can directly conclude that {\textsc{Independent Vertex Cover}} on $G$ has no solution.

So we may assume that $G = (V,E)$ is a bipartite graph, and we proceed to construct a solution $S$ of minimum size.  For each connected component of $G$, we define $(B_1,B_2)$ as the unique bipartition of its vertex set such that $|B_1| < |B_2|$ and $B_1$ and $B_2$ are two independent sets. (If $|B_1| = |B_2|$, we arbitrarily choose $B_1$ being one of them and $B_2$ being the other one.) Note that $S$ cannot contain vertices in both $B_1$ and $B_2$, since in that case by connectivity there would exist an alternating path in $G$ with only the endvertices in $S$, and then either there is an edge between both endvertices (contradicting the fact that $S$ should be an independent set), or some edge in the path does not contain vertices in $S$ (contradicting the fact that $S$ should be a vertex cover). Thus, if $S$ is a minimum-size solution, necessarily $S \cap (B_1 \cup B_2) = B_1$. (If $|B_1|=|B_2|$, then $S \cap (B_1 \cup B_2)$ is equal to either $B_1$ or $B_2$, and we assume without loss of generality that the former case holds.) Therefore, we start with $S = \emptyset$, and for each connected component of $G$, we add each element of $B_1$ to  $S$. After exploring the whole graph, if $|S| \leq k$ then we return $S$, otherwise we report that no such a set exists.\end{proof}

We would like to note that Theorem~\ref{thm:easy} in Section~\ref{sec:irlvd} generalizes Lemma~\ref{lemma:ivc} above.



The following simple lemma will be exhaustively used in the following sections. A problem $\Pi_1$ is \emph{polynomial-time reducible} to a problem $\Pi_2$ if there exists a polynomial-time algorithm that transforms an instance $I_1$ of $\Pi_1$ into an instance $I_2$ of $\Pi_2$ such that $I_1$ is a \textsc{Yes}-instance if and only if $I_2$ is.

\begin{lemma}
\label{lemma:rl}
Let $r$ and $\ell$ be two positive integers. Then
\begin{itemize}
\item[(i)] \textsc{$(r,\ell)$-Vertex Deletion} is polynomial-time reducible to \textsc{$(r,\ell+1)$-Vertex Deletion},
\item[(ii)] \textsc{$(r,\ell)$-Vertex Deletion} is polynomial-time reducible to \textsc{$(r+1,\ell)$-Vertex Deletion},
\item[(iii)] \textsc{Independent $(r,\ell)$-Vertex Deletion} is polynomial-time reducible to \textsc{Independent $(r,\ell+1)$-Vertex Deletion}, and
\item[(iv)] \textsc{Independent $(r,\ell)$-Vertex Deletion} is polynomial-time reducible to \textsc{Independent $(r+1,\ell)$-Vertex Deletion}.
\end{itemize}
Furthermore, in each of the above reductions the parameter remains unchanged. 
\end{lemma}

\begin{proof}
let $r$, $\ell$, and $k$ be three positive integers, and let $(G=(V,E),k)$ be an instance  of \textsc{$(r,\ell)$-Vertex Deletion} (for claims (i) and (ii)) or of \textsc{Independent $(r,\ell)$-Vertex Deletion} (for claims (iii) and (iv)), respectively.

\emph{Claim (i)}: Let $G' = (V',E')$ such that $V' = V \cup Q$, with
$Q$ a set of $r+k+1$ vertices, disjoint from $V$, and $E' = E \cup \{\{x,y\} : x,y \in Q, x \not = y\}$.
That is, $G'$ is the disjoint union of $G$ and a clique of size $(r+k+1)$.
Let $\mA$ be an algorithm solving \textsc{$(r,\ell+1)$-Vertex Deletion} on $(G',k)$ in time $f(n,k)$, for some given function $f$.

Assume first that $S$ is a solution given by $\mA$.
Then $G' \bs S$ is an $(r,\ell+1)$-graph.
Let $\mR'$ be the set of the $r$ independent sets of $G' \bs S$ and let $\mL'$ be the set of the $\ell+1$ cliques.
We claim that at least one clique of $\mL'$ is completely contained in $Q$.
Indeed, as $Q$ is a clique, then each set in $\mR'$ constains at most one vertex of $Q$.
As $|S| \leq k$ and $|Q|=k+r+1$, at least one vertex of $Q$ is contained in one of the cliques in $\mL'$.
Let $K$ be such a clique. As there are no edges between the vertices of $Q$ and the other vertices of $G'$, it follows that $K \subseteq Q$, as we wanted to prove. Thus, $G \bs (S \cap V)$ is an $(r,\ell)$-graph, and therefore  $|S \cap V|$ is a solution of the required size.

Conversely, if we can find $S$ a solution of \textsc{$(r,\ell)$-Vertex Deletion} on $(G,k)$, then $S$ is also a solution of \textsc{$(r,\ell+1)$-Vertex Deletion} on $(G',k)$.

Summarizing,  \textsc{$(r,\ell)$-Vertex Deletion} we can solved in time $f(n+r+k+1,k)$.

\emph{Claim (ii)}: As solving \textsc{$(r,\ell)$-Vertex Deletion} on $(G,k)$ is equivalent to solving \textsc{$(\ell,r)$-Vertex Deletion} on $(\overline{G},k)$, where $\overline{G}$ is the complement of $G$,  we can apply  claim (i) on $\overline{G}$.
Thus, if  \textsc{$(r+1,\ell)$-Vertex Deletion} can be solved in time $f(n,k)$, then \textsc{$(r,\ell)$-Vertex Deletion} can be solved in time $f(n+\ell+k+1,k)$.

\emph{Claim (iii)}: We follow the proof of claim~(i), but in this case, as $S$ is an independent set, it follows that $|S \cap Q| \leq 1$, so we can use a clique $Q$ of size $r+2$ instead of $r+k+1$.
 Hence, if  \textsc{Independent $(r,\ell+1)$-Vertex Deletion} can be solved in time $f(n,k)$, then  \textsc{Independent $(r,\ell)$-Vertex Deletion} can be solved in time $f(n+r+2,k)$.

\emph{Claim (iv)}:  We follow again the proof of claim~(i), but we redefine $Q$ to be an independent set whose vertices are completely adjacent to $V$. It can be easily checked that both instances are equivalent.
Thus, if \textsc{Independent $(r+1,\ell)$-Vertex Deletion} can be solved in time $f(n,k)$, then   \textsc{Independent $(r,\ell)$-Vertex Deletion} can be solved in time $f(n+\ell+k+1,k)$.
\end{proof}

\section{Well-known properties of iterative compression}
\label{ap:iterative-compression}

As mentioned in the introduction, iterative compression has been successfully used to obtain efficient algorithms for a number of parameterized problems~\cite{RSV04,GK0MPRR15,KrWa14}. In a nutshell, the main idea of this technique is to reduce in {\sf FPT}-time  a problem to solving a so-called \emph{disjoint} version of it, where we assume that we are given a solution of size almost as small as the desired one, and that allows us to exploit the structure of the graph in order to obtain the actual solution, which is required to be disjoint from the given one. This technique usually applies to hereditary properties.

A graph property $\QQ$ is \emph{hereditary} if any subgraph of a graph that satisfies $\QQ$ also satisfies $\QQ$. Let $\QQ$ be a hereditary graph property.
We define the following two problems in order to state two general facts about the technique of iterative compression, which we use in Section~\ref{sec:rlvd}.

\paraprobl
{\textsc{$\QQ$-Vertex Deletion}}
{A graph $G=(V,E)$ and an integer $k$.}
{A set $S \subseteq V$ of size at most $k$ such that $G\bs S$ satisfies property $\QQ$, or a correct report that such a set does not exist.}
{$k$.}

\vspace{-.4cm}
\paraprobl
{\textsc{Disjoint $\QQ$-Vertex Deletion}}
{A graph $G = (V,E)$, an integer $k$, and a set $S \subseteq V$ of size at most $k+1$ such that $G\bs S$ satisfies property $\QQ $.}
{A set $S' \subseteq V\sm S$ of size at most $k$ such that $G \bs S'$ satisfies property $\QQ$, or a correct report that such a set does not exist.}
{$k$.}

The following two results are well-known (cf. for instance~\cite{Cygan15}) and commonly assumed when using  iterative compression. We include the proofs here for completeness.

\begin{lemma}
\label{lemma:disjoint}
If \textsc{Disjoint $\QQ$-Vertex Deletion} can be solved in {\sf FPT}-time, 
then \textsc{$\QQ$-Vertex Deletion} can also be solved in {\sf FPT}-time.
\end{lemma}

\begin{proof}
Let $\mA$ be an {\sf FPT} algorithm which solves \textsc{Disjoint $\QQ$-Vertex Deletion}. 
Let $G=(V,E)$ be a graph and $k$ be an integer.
We want to solve \textsc{$\QQ$-Vertex Deletion} on $(G,k)$.
Let $v_1, \ldots, v_n$ be an arbitrary ordering of $V$.
For each $i \in \{0,\ldots, n\}$, let $V_i$ denote the subset of vertices $\{v_1, \ldots , v_i \}$ and $G_i=G[V_i]$.
We iterate over $i$ from $1$ to $n$ as follows.
At the $i$-th iteration, suppose we have a solution $S_i \subseteq V_i$ of \textsc{$\QQ$-Vertex Deletion}  on $(G_i,k)$.
At the next iteration, we can define $S_{i+1} = S_i \cup \{v_{i+1}\}$.
Note that $S_{i+1}$ is a solution of \textsc{$\QQ$-Vertex Deletion} on $(G_{i+1},k+1)$.
If $S_{i+1}$ is of size at most $k$ then it is a solution of \textsc{$\QQ$-Vertex Deletion}  on $(G_{i+1},k)$.
Assume that $S_{i+1}$ is of size exactly $k+1$.
We guess a subset $S$ of $S_{i+1}$ and we look for a solution $W$ of \textsc{$\QQ$-Vertex Deletion}  on $(G_{i+1},k)$ that does not contain any element of $S$.
For this, we use algorithm $\mA$  on $(H,|S|-1,S)$ with $H = G_{i+1} \bs (S_{i+1} \sm S)$.
If $\mA$ returns a solution $W$ then observe that the set $W \cup (S_{i+1} \sm S)$ is a solution of \textsc{$\QQ$-Vertex Deletion}  on $(G_{i+1},k)$.
If $\mA$  on $(H,|S|-1,S)$ does not return a positive answer for any of the possible guesses of $S$,
then \textsc{$\QQ$-Vertex Deletion}  on $(G_{i+1}, k)$ has no solution.
Since the property $\QQ$ is hereditary, \textsc{$\QQ$-Vertex Deletion}  on $(G, k)$ has no solution either, and therefore the algorithm returns that there is no solution. Thus, we obtain an algorithm solving \textsc{$\QQ$-Vertex Deletion} in {\sf FPT}-time, as we wanted.\end{proof}

\vspace{-.15cm}
\begin{corollary}
\label{coro:disjoint}
If \textsc{Disjoint $\QQ$-Vertex Deletion} can be solved in time $c^k\cdot n^{O(1)}$ for some constant $c$,
then \textsc{$\QQ$-Vertex Deletion} can be solved in time $(c+1)^k \cdot n^{O(1)}$.
\end{corollary}\vspace{-.15cm}
\begin{proof}
Let us argue about the running time of the algorithm of Lemma~\ref{lemma:disjoint}, assuming \textsc{Disjoint $\QQ$-Vertex Deletion} can be solved in time $c^k\cdot n^{O(1)}$ for some constant $c$.
The time required to execute $\mA$ for every subset $S$ at the $i$-th iteration is $\sum_{i=0}^{k+1} {k+1 \choose i} \cdot c^i\cdot n^{O(1)} = (c+1)^{k+1}\cdot n^{O(1)}$.
We obtain an algorithm that computes $\PP(\QQ)$ in time $(c+1)^k \cdot n^{O(1)}$, as we wanted.
\end{proof}

\section{$(r,\ell)$-\textsc{Vertex Deletion}}
\label{sec:rlvd}
%
%

By Lemma~\ref{lemma:rl}, in this section we may focus on the algorithm for \textsc{$(2,2)$-Vertex Deletion}. As we  use the technique of iterative compression, we need to define and solve the \emph{disjoint} version of the
$(2,2)$-\textsc{Vertex Deletion} problem. Indeed, if we solve the disjoint version, we just need to apply Corollary~\ref{coro:disjoint} in Section~\ref{ap:iterative-compression} to obtain a single-exponential {\sf FPT}-algorithm for $(2,2)$-\textsc{Vertex Deletion}.

\vspace{-.25cm}
\paraprobl
{\textsc{Disjoint $(r, \ell)$-Vertex Deletion}}
{A graph $G = (V,E)$, an integer $k$, and a set $S \subseteq V$ of size at most $k+1$ such that $G\bs S$ is an $(r, \ell)$-graph.}
{A set $S' \subseteq V\sm S$ of size at most $k$ such that $G \bs S'$ is an $(r, \ell)$-graph, or a correct report that such a set does not exist.}
{$k$.}


%
%

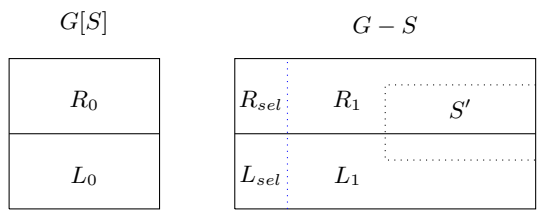
\begin{figure}
  \centering
  \begin{tikzpicture}
    \capt{3.5,0.2}{$L_1$};
    \capt{3.5,1.2}{$R_1$};
    \capt{0,0.2}{$L_0$};
    \capt{0,1.2}{$R_0$};
    \capt{5,1.1}{$S'$};

    \capt{0, 2.2}{$G[S]$};
    \capt{4,2.2}{$G \bs S$};

\add{
    \capt{2.35,1.2}{$\rs$};
    \capt{2.35,0.2}{$\ls$};
}

    \draw (0,0)
    -- +(0,2)
    -- +(2,2)
    -- +(2,0)
    -- +(0,0);

    \draw (3,-0)
    -- +(0,2)
    -- +(4,2)
    -- +(4,0)
    -- +(0,0);
\add{
  \draw[dotted] (3.7,0) -- +(0,2);
}

    \draw[dotted] (7,1.65) -- + (-2,0) -- +(-2,-1) -- +(0,-1);

    \draw (3,1) -- +(4,0);
    \draw (0,1) -- +(2,0);
  \end{tikzpicture}
  \caption{An $(r, \ell)$-partition of $G[S]$ and $G \bs S$ to solve \textsc{Disjoint $(r, \ell)$-Vertex Deletion} in the proof of Theorem~\ref{lemma:dttvd}.\vspace{-.0cm}}
  \label{fig:CS}
\end{figure}

}

\begin{theorem}
\label{lemma:dttvd}
{\textsc{Disjoint $(2,2)$-Vertex Deletion}} can be solved in time $2.31^{k} \cdot n^{O(1)}$, and therefore {\textsc{$(2,2)$-Vertex Deletion}} can be solved in time $3.31^{k} \cdot n^{O(1)}$.
\end{theorem}

\begin{proof}
Let $G$ be a graph, let $k$ be an integer, and let $S \subseteq V$ be a set of size at most $k+1$ such that $G \bs S$ is a $(2,2)$-graph.
We want to find a set $S' \subseteq V \sm S$ such that $G \bs S'$ is a $(2,2)$-graph with $|S'| \leq k$.
As the property of being a $(2,2)$-graph is hereditary, we can assume that $G[S]$ is a $(2,2)$-graph.
If it is not the case, we clearly have a \textsc{No}-instance and we stop.
We  guess a $(2,2)$-partition ($R_0$,$L_0$) of the graph $G[S]$,
and we  fix a $(2,2)$-partition \supprOK{($R_1$,$L_1$)}\addOK{$(R_S, L_S)$} of the graph $G\bs S$.
We guess $\ls \subseteq R_S$  and $\rs \subseteq L_S$, both of size at most $4$.
We define $R_1 = R_S \cup \rs \sm \ls$ and $L_1 = L_S \cup \ls \sm \rs$.
By Lemma~\ref{lemma:rlpartition},  there are at most $O(k^8\cdot n^{8})$ choices for $R_0$, $L_0$, \supprOK{$R_1$, and $L_1$}\addOK{$\rs$, and $\ls$}.
For each choice, we look for a solution $S' = R' \cup L'$ of size at most $k$ such that $R' \subseteq R_1$ and  $L' \subseteq L_1$.
A representation of this selection is depicted in Fig.~\ref{fig:CS}.
We define $L'$ as a smallest subset of $L_1$ such that $G[L_0 \cup (L_1 \setminus L')]$ is a $(0,2)$-graph.
In order to find it, we apply $k+1$ times the algorithm for \textsc{Restricted OCT} from Theorem~\ref{theorem:roct} to ($\overline{G}[L_0 \cup L_1]$,$L_1$, $i$) for $i$ from $0$ to $k$.
If the algorithm does not return a solution with  input ($\overline{G}[L_0 \cup L_1]$,$L_1 $, $k$) then the choice of $R_0$, \supprOK{$R_1$}\addOK{$\rs$}, $L_0$, and \supprOK{$L_1$}\addOK{$\ls$} is wrong, and we move to the next choice.
Otherwise, let $i_0$ be the smallest value of $i$ for which the algorithm returns a solution, and let $L'$ be this solution.
We now look for a set $R'$ of size at most $k-i_0$.
We find it by applying \addOK{the algorithm for} \textsc{Restricted OCT} to  ($G[R_0 \cup R_1]$, $R_1$, $k-i_0$).
If for some guess of $R_0$, \supprOK{$R_1$} \addOK{$\rs$}, $L_0$, and \supprOK{$L_1$} \addOK{$\ls$}, the algorithm returns a solution $R'$, then we output  $S'=R' \cup L'$ as our solution.
Otherwise we return that there is no solution.

Let us now analyze the running time of the algorithm. For each of the $O(k^8\cdot n^{8})$ guesses, we find $L'$ by applying \addOK{the algorithm for}  \textsc{Restricted OCT} $k+1$ times, and then we find $R'$ by applying \addOK{the algorithm for} \textsc{Restricted OCT}.
By Lemma~\ref{lemma:roct}, the claimed running time follows.

Now let us argue about the correctness of the algorithm.
If it outputs a set $S'$, then, by construction of the algorithm, this set is a solution of {\textsc{Disjoint $(2,2)$-Vertex Deletion}}.
Indeed, $L_0 \cup L_1 \sm S$ is a $(0,2)$-graph and $R_0 \cup R_1 \sm S$ is a $(2,0)$-graph.
On the other hand, assume that the instance of {\textsc{Disjoint $(2,2)$-Vertex Deletion}} has a solution $S^*$.
Let $(R^*,L^*)$ be a $(2,2)$-partition of $G \bs S^*$.
Then the solution $S^*$ can be found by the algorithm for the guess $R_0 = S \cap R^*$ and $L_0 = S \cap L^*$.
For this choice of $R_0$ and $L_0$, let $(R_S, L_S)$ be the fixed $(2,2)$-partition of $G\bs S$.
Let $G^* = (V^*,E^*)$ be the graph $G \bs (S \cup S^*)$.
Then $(R_s \cap V^*, L_S \cap V^*)$ and $(R^* \cap V^*, L^* \cap V^*)$ are two $(2,2)$-partitions of $G^*$.
By Lemma~\ref{lemma:rlpartition}, we can find $\ls \subseteq R_S$ and $\rs \subseteq L_S$ both of size at most $4$
such that $L_1 \cap V^* = L^* \cap V^*$ and $R_1 \cap V^* = R^* \cap V^*$, with  $R_1 = R_S \cup \rs \sm \ls$ and $L_1 = L_S \cup \ls \sm \rs$.
For this choice of $R_0$, $L_0$, $\rs$, and $\ls$,  the solution $S^*$ gives a value $i_0^*$ such that
the algorithm for \textsc{Restricted OCT} applied to ($\overline{G}[L_0 \cup L_1]$,$L_1 $, $i_0^*$) and
{the algorithm for} \textsc{Restricted OCT} applied to  ($G[R_0 \cup R_1]$, $R_1 $, $k-i_0^*$) will both return a solution.
Thus, if $S^*$ is a solution of {\textsc{Disjoint $(2,2)$-Vertex Deletion}}, then there is at least one choice of $R_0$, $L_0$, $\rs$, and $\ls$ such that the algorithm returns a solution.
\end{proof}

By combining Lemma~\ref{lemma:rl} and  Theorem~\ref{lemma:dttvd}, we obtain the following corollary\footnote{It is worth mentioning that if one is interested in optimizing the degree of the polynomial function $n^{O(1)}$ of our algorithms, we could solve directly the cases $(1,2)$ and $(2,1)$. In fact, this was the case in the original version of the paper, and Lemma~\ref{lemma:rl} was added after a remark of one of the referees.}
\begin{corollary}
\label{coro:vd}
{\textsc{$(2,1)$-Vertex Deletion}} and {\textsc{$(1,2)$-Vertex Deletion}} can be solved in time $3.31^k\cdot n^{O(1)}$.
\end{corollary}


It is known that \textsc{$(2,0)$-Vertex Deletion}, also known as \textsc{OCT}, cannot be solved in time $2^{o(k)}\cdot n^{O(1)}$ unless the \textsc{ETH} fails~\cite{IPZ01,LSW12}.
By combining this result with Lemma~\ref{lemma:rl}, we obtain that the running times of Theorem~\ref{lemma:dttvd} and Corollary~\ref{coro:vd} are asymptotically best possible in terms of $k$ under \textsc{ETH}.

\begin{theorem}
\label{theorem:nsevd}
Unless the \textsc{ETH} fails, there is no algorithm running in time $2^{o(k)}\cdot n^{O(1)}$ for solving \textsc{$(2,1)$-Vertex Deletion}, {\textsc{$(1,2)$-Vertex Deletion}}, or {\textsc{$(2,2)$-Vertex Deletion}}.
\end{theorem}

\section{Independent $(r, \ell)$-Vertex Deletion}
\label{sec:irlvd}

In this section we consider \textsc{Independent $(r, \ell)$-Vertex Deletion}. \addOK{Recall that the problem consists in finding a solution of \textsc{$(r, \ell)$-Vertex Deletion} that induces an independent set.}  We first provide in Subsection~\ref{sec:easy-hard}  a \addOK{(classical)} complexity dichotomy for the problem. In Subsection~\ref{sec:Independent-algos} we present {\sf FPT}-algorithms for the cases $(2,1)$ and $(2,2)$.  For the sake of the presentation, we postpone the running time analysis of these algorithms to Subsection~\ref{sec:time}. As we will see, these running times strongly depend on the running time required by the algorithm of Marx \emph{et al.}~\cite{MOR10} to solve the case $(2,0)$, that is \textsc{Independent OCT}, whose  bottleneck is to solve \textsc{Independent Mincut}.



\subsection{Easy and hard cases}
\label{sec:easy-hard}



We first deal with the polynomially-solvable cases in Theorem~\ref{thm:easy} and then we present an {\sf NP}-hardness reduction for the other cases in Theorem~\ref{thm:NP-hard}.

\begin{theorem}\label{thm:easy}
Let $r \in \{0,1\}$ and $\ell \in  \{0,1,2\}$ be two fixed integers.
The {\textsc{Independent ($r$,$\ell$)-Vertex Deletion}} problem can be solved in polynomial time.
\end{theorem}
\begin{proof}
Let us first consider the case $r = 1$ and $\ell = 2$.
One can check in polynomial time whether $G$ is a $(2,2)$-graph~\cite{Bra96}.
If it is not, then {\textsc{Independent $(1,2)$-Vertex Deletion}} on $(G,k)$ has no solution.
So assume that $G$ is a $(2,2)$-graph. By Lemma~\ref{lemma:rlpartition},  there are $O(n^{8})$ $(2,2)$-partitions of $G$ that can be computed in polynomial time.
We guess a $(2,2)$-partition $(R,L)$ of $G$, and we aim at partitioning $R$ into two independent sets $R_1$ and $R_2$ such that $|R_2| \leq k$.
If \textsc{Independent Vertex Cover} on $(G[R],k)$ has a solution $S$, then $R_1 = R \setminus S$ and $R_2 = S$ is the partition we want, and we return $S$.
Note that by Lemma~\ref{lemma:ivc}, \textsc{Independent Vertex Cover} can be solved in linear time on the graph $G[R]$.
If \textsc{Independent Vertex Cover} does not return a solution for any of the guesses of $(R,L)$, we return that our problem has no solution.

Finally, applying Lemma~\ref{lemma:rl} we obtain that {\textsc{Independent ($r$,$\ell$)-Vertex Deletion}} problem can be solved in polynomial time for every $r \in \{0,1\}$ and $\ell \in \{0,1,2\}$.\end{proof}



\begin{theorem}\label{thm:NP-hard}
Let $\ell \in  \{0,1,2\}$ be a fixed integer. The
{\textsc{Independent (2,$\ell$)-Vertex Deletion}} problem is {\sf NP}-hard.
\end{theorem}
\begin{proof} We first prove that \textsc{Independent $(2,0)$-Vertex Deletion} is {\sf NP}-hard.
We reduce from \textsc{$(2,0)$-Vertex Deletion}, commonly called \textsc{Odd Cycle Transversal}. 
The problem is proved to be {\sf NP}-complete in \cite{LeYa80}.

Let $G=(V,E)$ be a graph, let $k$ be an integer, and let $n = |V|$.
We want to solve \textsc{$(2,0)$-Vertex Deletion} on $(G,k)$.
We define $G' = (V',E')$, such that $V' = V \cup \{v_{e}^i, w_{e}^i :  e=\{v,w\} \in E, i \in \{0,\ldots, k\}\}$
and $E' = \{\{v,v_{e}^i\}, \{w,w_{e}^i\} : e=\{v,w\} \in E, i \in \{0 ,\ldots, k\}\} \cup \{\{v_{e}^{i},w_{e}^i\} : v \in V, w \in V, e=\{v,w\} \in E, i \in \{0 ,\ldots, k\}\}$.
That is, we replace each edge $e=\{v,w\}$ of $E$ by $n+1$ paths $v,v_e^i, w_e^i,w$ of length $3$, for $i \in \{0,\ldots, k\}$.
Assume we have a solution $S \subseteq V$ of  \textsc{$(2,0)$-Vertex Deletion} on $(G,k)$.
In $G'$, there is no edge between two vertices of $V$.
So $S$ is also a solution of \textsc{Independent $(2,0)$-Vertex Deletion} on $(G',k)$.
Now, we assume that $S$ is a solution of \textsc{Independent $(2,0)$-Vertex Deletion} on $(G',k)$.
We have that $S \cap V$ is also a solution of \textsc{Independent $(2,0)$-Vertex Deletion} on $G'$.
Indeed, assume $v_e^i \in V'$ is in $S$ for some $e = \{v,w\} \in E$ and $i \in \{0,\ldots, k\}$; the same analysis will apply to $w_e^i$.
If $v \in S$ then $v_e^i$ has only one neighbor in $G' \bs \{v\}$, so if $G' \bs S$ is bipartite, then so is $G'\bs (S \sm \{v_e^i\})$ and thus $S \sm \{v_e^i\}$ is also a solution.
If $w \in S$, then
 $v_e^i$ has only two neighbors in $G' \bs \{w\}$, namely $v$ and $w_e^i$, with $w_e^i$ being of degree $1$ in $G' \bs \{w\}$. So if $G' \bs S$ is bipartite, then so is $G'\bs (S \sm \{v_e^i\})$, and thus $S \sm \{v_e^i\}$ is also a solution.
So assume now that $v$ and $w$ are not in $S$.
Then there exists at least one index $i' \in \{0 ,\ldots, k\}$ such that $v_e^{i'}$ and $w_e^{i'}$ are not in $S$. This implies  that in the bipartite graph $G' \bs S$, $v$ and $w$ have to be on  opposite sides of the bipartition of $G'-S$.
We can safely add $v_e^i$ to $G\bs S$ such that the graph remains bipartite by adding $v_e^i$ to the side of the bipartition containing $w$.
So $S \sm \{v_e^i\}$ is also a solution.

By deleting all the vertices of the form $v_e^i$ from $S$, we obtain a set $S'$ such that $S' \subseteq V$ and $|S'| \leq k$.
As we preserve the property in $G'$ that if $\{v,w\} \in E$, with $v$ and $w$ not in $S'$, then $v$ and $w$ should be on opposite sides of the bipartition of  $G' \bs S'$, we have that $S'$ is a solution of \textsc{$(2,0)$-vertex deletion} on $G$.
This concludes the proof.


The {\sf NP}-hardness of {\textsc{Independent (2,1)-Vertex Deletion}} and {\textsc{Independent (2,2)-Vertex Deletion}} follow from the {\sf NP}-hardness of {\textsc{Independent (2,0)-Vertex Deletion}} by applying Lemma~\ref{lemma:rl}. \end{proof}


\subsection{{\sf FPT}-algorithms}
\label{sec:Independent-algos}

We deal with the cases $(2,2)$ and  $(2,1)$  in Theorem~\ref{theorem:ittvd} and  Corollary~\ref{theorem:itovd}, respectively.


\begin{figure}[t]

  \centering
  \begin{tikzpicture}
    \capt{4.7,-0.6}{$L_1$};
    \capt{4.7,0.4}{$R_1$};
    \capt{0,0.2}{$I$};
    \capt{0,1.2}{$R_0$};
    \capt{0,-0.85}{$L_0$};
    \capt{1.5,0.2}{$S'$};

\add{
\capt{5.7, -0.6}{$\ls$};
\capt{5.7, 0.4}{$\rs$};
}

    \capt{0, 2.2}{$G[S]$};
    \capt{4,2.2}{$G \bs S$};

    \draw (0,-0)
    -- +(0,2)
    -- +(2,2)
    -- +(2,0)
    -- +(0,0);

\add{
  \draw[dotted] (6.3,-0.8) -- +(0,2);
}

\draw (0,-0) -- +(0,-1) -- +(2,-1) -- +(2,0) ;

\draw[dotted] (-0.2,1) -- +(5,0) -- +(5,-1) -- +(0,-1) -- +(0,0);
    \draw (3,-0.8)
    -- +(0,2)
    -- +(4,2)
    -- +(4,0)
    -- +(0,0);

    \draw (3,0.2) -- +(4,0);
    \draw (0,1) -- +(2,0);
  \end{tikzpicture}

  \caption{An $(r, \ell)$-partition of $G[S\sm I]$ and $G \bs S$ to solve \textsc{Independent $(r, \ell)$-Vertex Deletion}.\vspace{-.0cm}}
  \label{fig:CSS}
\end{figure}
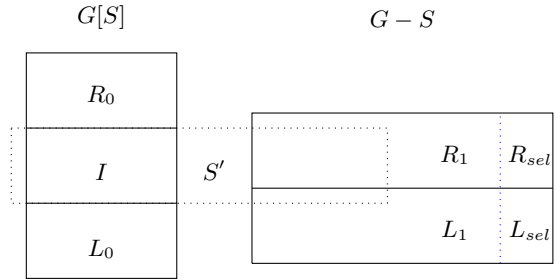

\begin{theorem}\label{theorem:ittvd}
{\textsc{Independent $(2,2)$-Vertex Deletion}} is {\sf FPT}.
\end{theorem}
\begin{proof}
The proof uses similar ideas than the proof of Theorem~\ref{lemma:dttvd}.
Let $G = (V,E)$ be a graph and let $k$ be an integer.
Let $S$ be a solution of the \textsc{$(2,2)$-vertex Deletion} problem on $(G,k)$.
Theorem~\ref{lemma:dttvd} gives us in {\sf FPT} time such a set $S$, or a report that such a set does not exist.
If there is no solution for \textsc{$(2,2)$-Vertex Deletion}, then {\textsc{Independent ($2$,$2$)-Vertex Deletion}} has no solution either.
So we can assume that such a set  $S$ exists.
Using $S$, we proceed to construct a solution $S'$ of our problem as follows.
We first guess an independent subset $I$ of $S$.
We want to construct a solution $S'$ of \textsc{Independent $(2,2)$-Vertex Deletion} such that $I \subseteq S'$ and $S' \cap S = I$.
If $G[S \sm I]$ is not a $(2,2)$-graph, then our choice of $I$ is wrong.
So assume $G[S \sm I]$ is a $(2,2)$-graph. We guess a  $(2,2)$-partition $(R_0, L_0)$  of $G[S\sm I]$, and we fix a $(2,2)$-partition $(R_S, L_S)$ of $G \bs S$.
We guess $\ls \subseteq R_S$ of size at most $4$ and $\rs \subseteq L_S$ of size at most $4$.
We define $R_1 = R_S \cup \rs \sm \ls$ and $L_1 = L_S \cup \ls \sm \rs$.
By Lemma~\ref{lemma:rlpartition},  there are at most $O(k^8\cdot n^{8})$ choices for $R_0$, $L_0$, \supprOK{$R_1$, and $L_1$}\addOK{$\rs$, and $\ls$}.
We want to find $R' \subseteq R_1$ and $L' \subseteq L_1$ such that $S' = I \cup L' \cup R'$.
A representation of this selection is depicted  in Fig.~\ref{fig:CSS}.
As we want the solution to induce an independent set, at most two elements of $L_1 $ are in $S'$, that is, $|L'| \leq 2$.
We guess these at most two vertices that define $L'$ such that $L' \cup I$ is an independent set and such that $G[(L_0 \cup L_1)\sm L']$ is a $(0,2)$-graph.
If it is not the case then our choice is wrong.
We now have to find $R'$ of size at most $k - |I| - |L'|$.
For this, we apply \addOK{the algorithm for} {\textsc{Restricted Independent OCT}} on $(G[R_0 \cup R_1],D,k-|I| - |L'|)$ with $D = \{x \in R_1  : \forall y \in I \cup L' , \{x,y\} \not \in E\}$.
If it returns a solution $R'$ then we can output the solution $S'=I \cup L' \cup R'$.
If it does not return a solution for any of the guesses of $I$, $R_0$, $L_0$, \supprOK{$R_1$, and $L_1$} \addOK{$\rs$, and $\ls$}, then we return that there is no solution of \textsc{Independent $(2,2)$-Vertex Deletion}.

Now let us argue about the correctness of the algorithm.
If it outputs a set $S'$, then, by construction of the algorithm this set is a solution of {\textsc{Independent $(2,2)$-Vertex Deletion}}.
Indeed, $L_0 \cup L_1 \sm L'$ is a $(0,2)$-graph, $R_0 \cup R_1 \sm R'$ is a $(2,0)$-graph, and $S'$ is an independent set.
Now, assume that our instance of {\textsc{Independent $(2,2)$-Vertex Deletion}} has a solution $S^*$.
Let $(R^*,L^*)$ be an $(2,2)$-partition of $G \bs S^*$.
Then the solution $S^*$ can be found by the algorithm for the guess $I = S \cap S^*$, $R_0 = S \cap R^*$, and $L_0 = S \cap L^*$.
For this choice of $R_0$ and $L_0$, let $(R_S, L_S)$ be the fixed $(2,2)$-partition of $G\bs S$.
Let $G^* = (V^*,E^*)$ be the graph $G \bs (S \cup S^*)$.
Then $(R_s \cap V^*, L_S \cap V^*)$ and $(R^* \cap V^*, L^* \cap V^*)$ are two $(2,2)$-partitions of $G^*$.
By Lemma~\ref{lemma:rlpartition}, we can find $\ls \subseteq R_S$ and $\rs \subseteq L_S$ both of size at most $4$
such that $L_1 \cap V^* = L^* \cap V^*$ and $R_1 \cap V^* = R^* \cap V^*$, with  $R_1 = R_S \cup \rs \sm \ls$ and $L_1 = L_S \cup \ls \sm \rs$.
For this choice of $R_0$, $L_0$, $\rs$, and $\ls$, the algorithm will  define $L' = S^* \cap L_1$.
Then the existence of the solution $S^*$ certifies that
{the algorithm for} {\textsc{Restricted Independent OCT}} on $(G[R_0 \cup R_1],D,k-|I| - |L'|)$,
 with $D = \{x \in R_1  : \forall y \in I \cup L' , \{x,y\} \not \in E\}$,
will return a solution.
Thus, if $S^*$ is a solution of {\textsc{Independent $(2,2)$-Vertex Deletion}}, then there is at least one choice $I$, $R_0$, $L_0$, $\rs$, $\ls$, and $L'$ such that the algorithm returns a solution.
\end{proof}

By combining Lemma~\ref{lemma:rl} and Theorem~\ref{theorem:ittvd}, we obtain the following corollary.
\begin{corollary}
\label{theorem:itovd}
{\textsc{Independent $(2,1)$-Vertex Deletion}} is {\sf FPT}, with the same running time as {\textsc{Independent $(2,2)$-Vertex Deletion}}. 
\end{corollary}

\subsection{Analysis of the running time}
\label{sec:time}

In this subsection we provide an upper bound on the running times of the {\sf FPT}-algorithms  for  \textsc{Independent $(2,2)$-Vertex Deletion} and \textsc{Independent $(2,1)$-Vertex Deletion} given by Theorem~\ref{theorem:ittvd} and Corollary~\ref{theorem:itovd}, respectively.
Note that in these algorithms, the only non-explicit running time is the one of the algorithm for \textsc{Restricted Independent OCT} given by Theorem~\ref{theorem:roct}. To obtain this upper bound, we will go through the main ideas of the algorithm of Marx \emph{et al.} \cite{MOR10} for \textsc{Independent OCT}, and then by using the same tools used in  the proof of Lemma~\ref{lemma:roct} we will obtain the same upper bound for the restricted version of \textsc{Independent OCT}.

We need to define the following problem, where an \emph{$s-t$ cut} in a graph $G$ is a set of vertices $C$ such that  $s$ is not connected to $t$ in the graph $G - C$.


\paraprobl
{\textsc{Independent Mincut}}
{A graph $G = (V,E)$, an integer $k$, and two vertices $s,t \in V$.}
{An $s-t$ cut $C \subseteq V\sm \{s,t\}$ such that $|C| \leq k$ and $C$ is an independent set, or a correct report that such a set does not exist.}
{$k$.}

We provide here a sketch of proof of the following simple lemma. We first recall for completeness the definition of treewidth.
A \emph{tree-decomposition} of width $w$ of a graph $G=(V,E)$ is a pair $(T,\sigma)$, where $T$ is a tree and $\sigma = \{ B_t : B_t \subseteq V, t \in V(T) \}$ such that:
\begin{itemize}
\item[$\bullet$] $\bigcup_{t \in V(T)} B_t = V$,
\item[$\bullet$] For every edge $\{u,v\} \in E$ there is a $t \in V(T)$ such that $\{u, v\} \subseteq B_t$,
\item[$\bullet$] $B_i \cap B_k \subseteq B_j$ for all $\{i,j,k\} \subseteq V(T)$ such that $j$ lies on the path $i ,\dots, k$ in $T$, and
\item[$\bullet$] $\max_{i \in V(T)} |B_t| = w +1$.
\end{itemize}

The sets $B_t$ are called \emph{bags}. The \emph{treewidth} of $G$, denoted by ${\bf tw}(G)$, is the smallest integer $w$ such that there is a tree-decomposition of $G$ of width $w$. An \emph{optimal tree-decomposition} is a tree-decomposition of width ${\bf tw}(G)$.

\begin{lemma}
\label{lemma:imtw}
\textsc{Independent Mincut} can be solved in time $3^{{\bf tw}}\cdot n^{O(1)}$, where ${\bf tw}$ stands for the treewidth of the input graph.
\end{lemma}
\begin{proof}
(Sketch)
For each bag $B$ of the tree-decomposition, we store all quadruples $(S,T,D,\ell)$ such that we  have already found a cut $C'$ of size at most $\ell$ in the explored graph such that $B \cap C' = D$, such that there is no edge between $S$ and $T$, and such that $s \in B$ if and only if $s \in S$ and $t \in B$ if and only if $t \in T$ .
There are at most $3^{{\bf tw}}\cdot k$ such quadruples, and so the lemma follows.
\end{proof}

Note that Lemma~\ref{lemma:imtw} is a special case of the problems covered by the result of Pilipczuk~\cite{Pili11}, who proves that such problems can be solved in time $c^{{\bf tw}} \cdot n^{O(1)}$ for some constant $c >0$.
We also need the following result, where the key idea is to obtain an {\sl equivalent} graph whose treewidth is bounded by a function of $k$.

\begin{theorem}[Marx \emph{et al.}~\cite{MOR10}]
\label{theorem:trt}
Let $G = (V,E)$ be a graph, let $S \subseteq V(G)$, and let $k$ be an integer.
Let $C$ be the set of all vertices of $G$ participating in a minimal $s-t$ cut of size at most $k$ for some $s,t \in S$.
Then there is an algorithm running in time $2^{O(k^2)} \addOK{\cdot |S|^2} \cdot n^{O(1)}$ that computes a graph $G^*$ and a tree-decomposition of $G^*$ of width at most  $2^{O(k^2)} \cdot |S|^2$  having the following properties:
\begin{itemize}
\item[$\bullet$] $C \cup S \subseteq V(G^*)$,
\item[$\bullet$] For every $s,t \in S$, a set $K \subseteq V(G^*)$ with $|K| \leq k$ is a minimal $s-t$ cut of $G^*$ if and only if $K \subseteq C \cup S$ and $K$ is a minimal $s-t$ cut of $G$,
\item[$\bullet$] The treewidth of $G^*$ is at most $2^{O(k^2)}\addOK{\cdot |S|^2}$, and
\item[$\bullet$] For any $K \subseteq C$, $G^*[K]$ is isomorphic to $G[K]$.
\end{itemize}
\end{theorem}

\begin{lemma}\label{lem:Indep-Mincut}
\textsc{Restricted Independent Mincut} can be solved in time $2^{2^{O(k^2)}}\cdot n^{O(1)}$.
\end{lemma}

\begin{proof}
We first deal with \textsc{Independent Mincut}. Let $G = (V,E)$ be a graph and $k$ be an integer.
Let $G^*$ be the graph satisfying the requirements of Theorem~\ref{theorem:trt} for $S = \{s,t\}$.
Following the proof of \cite[Theorem 3.1]{MOR10}, it follows that $(G,s,t,k)$ has a solution of \textsc{Independent Mincut} if and only if $(G^*,s,t,k)$ has one.
We can now apply Lemma~\ref{lemma:imtw}, and solve \textsc{Independent Mincut} on $(G^*,s,t,k)$ in time $3^{{\bf tw}(G^*)}\cdot n^{O(1)}$.
By Theorem~\ref{theorem:trt}, $\tw(G^*) = 2^{O(k^2)}$ and the lemma follows.

Finally, note that the restricted version of \textsc{Independent Mincut} can be solved within the same running time by making enough copies of each ``undesired'' vertex, as in the proof of Lemma~\ref{lemma:roct}.
\end{proof}



\begin{theorem}
\textsc{Restricted Independent OCT} can be solved in time $2^{2^{O(k^2)}}\cdot n^{O(1)}$.
\end{theorem}

\begin{proof} In the following, we give a proof for  \textsc{Independent OCT}. In order to obtain the algorithm for \textsc{{Restricted} Independent OCT}, we just take again, as we did for \textsc{Restricted OCT} in Lemma~\ref{lemma:roct}, a larger instance by making enough copies of each ``undesired'' vertex, and apply \textsc{Independent OCT}.

Let $G=(V,E)$ be a graph and let $X$ be a solution of \textsc{OCT} on $(G,k)$, which we can assume to exist. Let $(S_1, S_2)$ be a partition of $G \bs X$ into two independent sets.
We define an auxiliary graph $G' = (V', E')$ as defined in \cite{RSV04}.
So we have $V' = (V \sm X) \cup \{x_1,x_2 : x \in X\}$ and
$E' = \{\{v,w\} \in E : v,w \in V\setminus X\} \cup \{\{y,x_{3-i}\} : y \in S_i, x \in X, i \in \{1,2\}, \{x,y\} \in E\} \cup \{\{x_1,y_2\}, \{x_2,y_1\} : x,y \in X, \{x,y\} \in E\}$.
Given $Y \subseteq X$, we say that a partition of $Y'= \{y_1,y_2 : y \in Y\}$ into two sets ($Y_A$, $Y_B$) is \emph{valid} if for all $ y \in Y$, exactly one of $y_1,y_2$ is in $Y_A$.
We let $S = S_1 \cup S_2$.

\vspace{.2cm}
To continue, we need the next reformulation of \cite[Lemma 1]{RSV04} and its proof.
\begin{claimN}\label{claim:1}
There is an independent odd cycle transversal $Z$ of size at most $k$ in $G$ if and only if there exists $Y\subseteq X$ and a valid partition $(Y_A,Y_B)$ of $Y'$ such that there is an independent mincut $C\subseteq S$ that separates 
$Y_A$ from 
$Y_B$ in $G'$ and such that $Z = C \cup (X \setminus Y)$ is an independent set of size at most $k$ in $G$.
\end{claimN}
\begin{proof} ($\Rightarrow$)
Let $Z$ be an independent odd cycle transversal of size at most $k$ in $G$.
We assume that $Z$ is of minimum size and that its removal produces two independent sets $S^Z_1$ and $S^Z_2$. Let $K = Z \cap X$, let ${J} = Z \setminus K$, and let $Y = X \setminus K$.
We define $(Y_A,Y_B)$, a valid partition of $Y'$ such that
$Y_A = \{y_1 : y \in Y \cap S_1^Z\} \cup \{y_2 : y \in Y \cap S_2^Z\}$ and $Y_B = \{y_2 : y \in Y \cap S_1^Z\} \cup \{y_1 : y \in Y \cap S_2^Z\}$.

We claim that $J$ is a cutset of $G'[Y_A\cup Y_B \cup S]$ separating $Y_A$ from $Y_B$.
Take a minimal path $P$ from $Y_A$ to $Y_B$ in $G'[(Y_A \cup Y_B \cup S) \setminus  J]$.
Let $u$ and $v$ be the two endpoints of $P$.
By minimality of $P$, $P \cap (Y_A \cup Y_B) = \{u,v\}$.
We assume without loss of generality that either $u,v \in Y \cap S_1^Z$ or $u \in Y \cap S_1^Z$ and $v \in Y \cap S_2^Z$.
In the former case, we have that $u = y_1$ for some $y \in Y$ and $v = w_2$ for some $w \in Y$.
As, by construction, $G'[Y_A\cup Y_B \cup S]$ is bipartite and $y_1$ and $w_2$ are on opposite sides of the bipartition of $G'[(Y_A\cup Y_B \cup S)\sm J]$,
necessarily $P$ has odd size.
But as $u,v \in S_1^Z$,
$u$ and $v$ are on the same side of the bipartition of $G \bs Z$, and so of $G'[(Y_A\cup Y_B \cup S)\sm J]$ as well.
Thus $P$ should have even size, a contradiction.
We obtain a similar contradiction in the latter case.

($\Leftarrow$) Under the condition of  Claim~\ref{claim:1}, $Z$ is an independent set, and we need to prove that it is an odd cycle transversal as well.
Assume that there is an odd cycle $O$ in $G \bs Z$. Then by definition of $X$, $O$ intersects $X$ at least once.
Let $O^0, \ldots, O^{m-1}$ be the $m$ times $O$ intersects $X$ and we define $O^m = O^0$.
We have that $O^i \not = O^j$ for all $i <j < m$.
For each $i \in \{0,\ldots, m-1\}$, let $P_i$ be the path from $O^i$ to $O^{i+1}$.
As $O$ never intersects $Z$, then in $G'$ the path $P_i$ never goes from $Y_A$ to $Y_B$.
It means that for each $i$ such that $P_i$ is of even size, $O^{i}_1$ and $O^{i+1}_1$ are in the same set $Y_A$ or $Y_B$,
and for each $i$ such that $P_i$ is of odd size, $O^{i}_1$ and $O^{i+1}_2$ are in the same set $Y_A$ or $Y_B$.
But $O$ is an odd cycle, so there is an odd number of paths $P_i$ such that $P_i$ is of odd size.
We deduce that such an odd cycle $O$ cannot exist, implying that $Z$ is an independent odd cycle transversal. \end{proof}

We now apply \addOK{the algorithm for} \textsc{Independent Mincut} for all $Y \subseteq X$ and all valid partitions $(Y_A,Y_B)$ of $Y'$.
Note that we need to consider a restricted version of \textsc{Independent Mincut} because we do not want the neighborhood of $X \setminus Y$ to be in the solution.
By Claim~\ref{claim:1}, if we obtain a solution, then we have found our independent odd cycle transversal, and otherwise we can safely return that such a set does not exist. The claimed running time follows from Lemma~\ref{lem:Indep-Mincut}.
\end{proof}

By using the same argument of Lemma~\ref{lemma:roct}, we obtain the following corollary.

\begin{corollary}\label{cor:run-time}
\textsc{Restricted Independent OCT} can be solved in time $2^{2^{O(k^2)}}\cdot n^{O(1)}$, and therefore \textsc{Independent $(r,\ell)$-Vertex Deletion} can also be solved in time $2^{2^{O(k^2)}}\cdot n^{O(1)}$ for $r= 2$ and $\ell \in \{0,1,2\}$.
\end{corollary}


Note that the previous results would be automatically improved if one could find a faster algorithm for \textsc{Independent Mincut}.

\vspace{.5cm}

\noindent \textbf{Acknowledgement}. We would like to thank the anonymous referees  for helpful remarks that improved and simplified the presentation of the manuscript.

\bibliographystyle{abbrv}
\bibliography{bib-dichotomy}

\end{document}